\documentclass{llncs}
\pdfoutput=1

\usepackage{graphicx}
\usepackage{amssymb}
\usepackage{amsmath,amsfonts}
\usepackage{epstopdf}
\usepackage[utf8]{inputenc}
\usepackage[english]{babel}
\usepackage{microtype}
\usepackage[T1]{fontenc}
\usepackage{authblk}
\usepackage{subfigure}
\usepackage{url}
\usepackage[backgroundcolor=yellow!20,textsize=scriptsize,textwidth=120pt]{todonotes}
\reversemarginpar

\newtheorem{algorithm}{Algorithm}
\newtheorem{observation}{Observation}

\renewcommand{\paragraph}[1]{\noindent\textbf{\emph{#1}}}

\def\longrightharpoonup{\relbar\joinrel\rightharpoonup}
\def\longleftharpoondown{\leftharpoondown\joinrel\relbar}
\def\longrightleftharpoons{\mathop{\vcenter{\hbox{\ooalign{\raise1pt\hbox{$\longrightharpoonup\joinrel$}\crcr\lower1pt\hbox{$\longleftharpoondown\joinrel$}}}}}}
\def\rxn{\mathop{\rightarrow}\limits}  

\newcommand{\bigo}{\mathcal{O}}

\title{Robust Detection in Leak-Prone Population Protocols}
\author{
Dan Alistarh\inst{1}\inst{5}\thanks{Supported by an SNF Ambizione Fellowship.} 
\and 
Bart\l{}omiej Dudek\inst{2}
\and 
Adrian Kosowski\inst{3}\thanks{Supported by Inria project GANG, ANR project
DESCARTES, and NCN grant 2015/17/B/ST6/01897.}
\and 
David Soloveichik\inst{4}\thanks{Supported by NSF grants CCF-1618895 and CCF-1652824.}
\and 
Przemys\l{}aw~Uzna\'nski\inst{1}
}

\institute{ETH Z\"urich, Switzerland \and
University of Wroc\l{}aw, Poland \and Inria Paris and IRIF, Universit\'e Paris Diderot, France \and University of Texas, Austin, TX, USA \and IST Austria\\\ \\E-mail addresses: \texttt{\{dan.alistarh,przemyslaw.uznanski\}@inf.ethz.ch},\texttt{bartlomiej.dudek@cs.uni.wroc.pl},\\\texttt{adrian.kosowski@inria.fr}, \texttt{david.soloveichik@utexas.edu}.}

\begin{document}

\maketitle

\setcounter{footnote}{0}
\begin{abstract}
In contrast to electronic computation,
chemical computation is noisy and susceptible to a variety of sources of error,
which has prevented the construction of robust complex systems.
To be effective, chemical algorithms must be designed with an appropriate error model in mind.
Here we consider the model of chemical reaction networks that preserve molecular count (population protocols),
and ask whether computation can be made robust to a natural model of unintended ``leak'' reactions.
Our definition of leak is motivated by both the particular spurious behavior seen when implementing chemical reaction networks with DNA strand displacement cascades,
as well as the unavoidable side reactions in any implementation due to the basic laws of chemistry.
We develop a new ``Robust Detection'' algorithm for the problem of fast (logarithmic time) single molecule detection,
and prove that it is robust to this general model of leaks.
Besides potential applications in single molecule detection, the error-correction ideas developed here might enable a new class of robust-by-design chemical algorithms. Our analysis is based on a non-standard hybrid argument, combining ideas from discrete analysis of population protocols with classic Markov chain techniques.

\end{abstract}

\section{Introduction}
A major challenge in designing autonomous molecular systems is to achieve a sufficient degree of error tolerance despite the error-prone nature of the chemical substrate.
While considerable effort has focused on making the chemistry itself more robust, here we look at the possibility of developing chemical algorithms that are inherently resilient to the types of error encountered.
Before designing robust chemical algorithms, we must decide on a good error model that is relevant to the systems we care about.
In this paper we focus on a very simple and general error model that is motivated both by basic laws of chemistry as well as by implementation artifacts in strand displacement constructions for chemical reaction networks. We begin by listing the types of errors we aim to capture. 

\vspace{5pt}

\paragraph{Leaks due to Law of Catalysis.}
A fundamental law of chemical kinetics is that for every catalyzed reaction, there is an uncatalyzed one that occurs at a (often much) slower rate.
By a \emph{catalytic} reaction, we mean a reaction that involves some species $X$ but does not change its amount; this species is called a \emph{catalyst} of that reaction.
For example, the reaction $X + Y \rxn X + Z$ is catalytic, and species $X$ is the catalyst since its count remains unchanged by the firing of this reaction.
By the law of catalysis, reaction $X + Y \rxn X + Z$ must be accompanied by a (slower) \emph{leak} reaction $Y\rxn Z$.
(A more general formulation of the law of catalysis is that if any sequence of reactions does not change the net count of $X$, then there is a pathway that has the same effect on all the other species, but can occur in the absence of $X$ (possibly much slower).
Thus for example, if $X + Y \rxn W$ and $W \rxn X + Z$ are two reactions, then there must also be a leak reaction $Y\rxn Z$.
Formally defining catalytic cycles and catalysts is non-trivial and is beyond the scope of this paper~\cite{gopalkrishnan2011catalysis}.)

\vspace{5pt}

\paragraph{Leaks due to Law of Reversibility.}
Another fundamental law of chemical kinetics is that any reaction occurs also in the reverse direction at some (possibly much slower) rate.
In other words, reaction $X + Y \rxn Z + W$ must be accompanied by $Z + W \rxn X + Y$.
(The degree of reaction reversibility is related to the free-energy use, such that irreversible reactions would require ``infinite'' free energy.)

\vspace{5pt}

\paragraph{Leaks due to Spurious Activation in Strand-Displacement Cascades.}
Arbitrary chemical reaction networks can in principle be implemented with DNA strand displacement cascades~\cite{soloveichik2010dna,cardelli2013two}.
Implementations based on strand displacement  also suffer from the problem of leaks~\cite{chen2013programmable}.
The implementation of a reaction like $X + Y \rxn Z + W$ consists of ``fuel'' complexes present in excess, that hold $Z$ and $W$ sequestered.
A cascaded reaction of the fuel complex with $X$ and $Y$ results in the release of active $Z$ and $W$.
Leaks in this case consist of $Z$ and $W$ becoming spuriously activated, even in the absence of active $X$ or $Y$.

Importantly, for a catalytic reaction such as $X + Y \rxn X + Z$, it is possible to design a strand displacement implementation that does not leak the catalyst $X$.
This implementation would release the same exact molecule of $X$ as was consumed to initiate the process, as in the catalytic system described in~\cite{zhang2007engineering}. Since fuels do not hold $X$ sequestered, $X$ cannot be produced in the leak process (although $Z$ can).

\vspace{5pt}

\paragraph{Modeling Reactions and Leaks.}
Note that in all cases above, we can guarantee that a species does not leak if it is exclusively a catalyst in every reaction it occurs in.
This allows us some handle on the leak.
In particular, we will ensure that the species we are trying to detect (called $D$ below) will be a catalyst in every reaction that involves it. 
Otherwise, there might be a leak pathway to generating $D$ itself---which is fundamentally irrecoverable.

We express the implementations below in the \emph{population protocol} formalism~\cite{angluin2006passivelymobile}.
That is, we consider a system with $n$ molecules (aka nodes), which interact uniformly at random in a series of discrete steps.
A population protocol is given as a set of reactions (aka transition rules) of the form
$$ A + B \rightarrow C + D.$$
Note that unlike general reaction networks, population protocols conserve total molecular count since molecules never combine or split.
For this reason, compared to general chemical reaction networks, this model is  easier to analyze.

Given the set of reactions defining a protocol, we partition the species into \emph{catalytic} states, which never change count as a consequence of any reaction, and \emph{non-catalytic}, otherwise.
Crucially, we model \emph{leaks} as spurious reactions which can consume and create arbitrary non-catalytic species.
More formally, a leak is a reaction of the type
$$ S \rightarrow S',$$

\noindent where $S$ and $S'$ denote arbitrary non-catalytic species.
In the following, we do not make any assumptions on the way in which these leak transitions are chosen (i.e., they could in theory be chosen \emph{adversarially}), but we assume an upper bound  on the \emph{rate} at which leaks may occur in the system, controlled by a parameter $\beta$. 

\paragraph{Leak-Robust Detection.}
A computationally simple task which already illustrates the difficulty of information processing in such an error-prone system is \emph{single molecule detection}.
Consider a solution of $n$ molecules, in which a single molecule $D$ may or may not be present.
Intuitively, the goal is to generate large-scale (in the order of $n$) change in the system, depending on whether or not $D$ is present or absent.
Our time complexity measure is \emph{parallel time}, defined as the number of pairwise interactions, divided by $n$. 
This measure of time naturally captures the parallelism of the system, where each molecule can participate in a constant number of interactions per unit time.
Subject to leaks, our goal is to design the chemical interaction rules (formalized as a population protocol) to satisfy the following behavior. 
If $D$ is present then it is detected fast, in logarithmic parallel time, 
and that the output is probabilistically ``stable'' in the sense that sampled at a random future time the system is in the ``detected configuration'' with high probability.
By contrast, if $D$ is absent, then the system sampled at a random future time should be in the ``undetected configuration'' with high probability.
This basic task has several variations, for instance signal amplification or approximate counting of $D$.

We first develop some intuition about this problem, by considering some strawman approaches.

A first trivial attempt would be to have neutral molecules become ``detectors'' (state $T$) as soon as they encounter $D$, that is,
$$ D + N \rxn D + T. $$
This approach suffers from two fatal flaws. First, it is \emph{slow}, in that detection takes \emph{linear} parallel time. Second, it has no way from recovering from leaks of the type $N \rxn T$.

A second attempt could try to implement an epidemic-style detection of $D$, that is:
$$ D + N \rxn D + T $$
$$ T + N \rxn T + T. $$

This approach is \emph{fast}, i.e. converges in \emph{logarithmic} parallel time in case $D$ is present. However, if $D$ is not present, the algorithm converges to a \emph{false positive} state: a leak of the type $N \rxn T$  brings the system to an all-$T$ state, despite the absence of $D$.
One could try to add a ``neutralization'' pathway by having $T$ turn back to $N$ after a constant number of interactions, but a careful analysis shows that
this approach also fails to recover from leaks of the type  $N \rxn T$.

Thus, it is not clear whether leak-resistant detection is possible in population protocols (or more generally chemical reaction networks).
There has been considerable work in the algorithmic community on diffusion based models, e.g.~\cite{DBLP:conf/focs/KarpSSV00}. 
However, such results do not seem to apply to this setting, since leak models have not been considered previously, and none of the known techniques are robust to leaks. In particular, it appears that techniques for \emph{deterministic} computation in population protocols do not carry over in the presence of leaks.
More generally, this seems to create an unfortunate gap between the algorithmic community, which designs and analyzes population protocols in leak-free models, and
more practically-minded research, which needs to address such implementation issues.

\paragraph{Contribution.}
In this paper, we take a step towards bridging this gap. We provide a general algorithmic model of leaks, and apply it to the detection problem.
Specifically, our immediate goal is to elucidate the question of whether efficient, leak-robust detection is possible.

We prove that the answer is yes. We present a new algorithm, called \emph{Robust-Detect}, which guarantees the following.
Assume that the rate at which leaks occur is upper bounded
 by $\beta / n  \ll 1 / n$, and that we return the output mapping (detect/non-detect) of a randomly chosen molecule after $O(\log n)$ parallel time.
Then the probability of a \emph{false negative} is at most $1 / e + o(1)$, and the probability of a \emph{false positive} is at most $\beta$.
(Note that as the total molecular count $n$ increases, the chance that a particular interaction involves $D$ decreases linearly with $n$.
Thus the leak rate must also decrease linearly with $n$, or else the leaks will dominate.
Alternatively, we can view some fixed leak rate as establishing an upper bound on the molecular count $n$, see below.
)

\paragraph{Algorithm Description.}  We now sketch the intuition behind the algorithm and its properties, leaving the formal treatment to Sections~\ref{sec:algo} and~\ref{sec:analysis}.
Fix a parameter $s \geq 1$, to be defined later. We define a set of ``detecting'' species $X_1, \ldots, X_s$, arranged in consecutive levels.
Whenever a molecule meets $D$, it moves to the highest ``alert'' level, $X_1$.
Since leaks might produce this species as well, we decay it gracefully across $s$ levels.
More precisely, whenever a molecule at level $X_i$ meets another molecule at level $X_j$, both molecules move to state
$X_{\min(i, j) + 1}$.
A molecule which would move beyond level $X_s$ following a reaction becomes neutral, i.e. moves to species $N$.
Nodes in state $X_i$ with $i < s$ turn $N$ into $X_{i + 1}$, whereas molecules in state $X_s$ also become neutral when interacting with $N$.

\paragraph{Analysis.} Intuitively, the algorithm's dynamics for the case where a single molecule is in state $D$ are as follows. The counts of molecules in state $X_i$ tend to increase \emph{exponentially} with the alert level $i$, up to levels $\approx \log n$, when the count becomes a constant fraction of $n$.
However, once level $\log n$ is reached, these counts decrease \emph{doubly exponentially}. Thus, it suffices to  set $s = \log n$ to obtain that a fraction of at least $(1 - 1 / e)$ molecules are in one of the alert states $X_i$ in case $D$ is present. It is not hard to prove that leaks cannot meaningfully affect the convergence behavior in this case.

The other interesting case is when $D$ is not present, but leaks may occur, leading to possible false positives.
Intuitively, we can model this case as one where states $X_1$ at the highest alert level simply are created at a lower rate $\beta / n \ll 1 / n$.
A careful analysis of this setting yields that the probability of a \emph{false positive} ($D$ detected, but not present) in this case is at most $\beta$, corresponding to the leak rate parameter.

Our analysis technique works by characterizing the stationary behavior of the Markov chain corresponding to the algorithm, and the convergence properties (mixing time) of this chain.
For technical reasons, the analysis uses a non-standard hybrid argument, combining ideas from discrete analysis of population protocols with classic Markov chain techniques.
The argument proves that the algorithm always stabilizes to the correct output in logarithmic parallel time.

The analysis further highlights a few interesting properties of the algorithm.
First, if the detectable species $D$ is present in a higher count $k > 1$, then the algorithm effectively skips the first $\log k$ levels, and thus requires $\log (n / k) + O( \log \log n )$ states.
Second, it is not necessary to know the exact value of $\log n$, as the counts of species past this threshold decrease doubly exponentially.

\paragraph{Alternative Formulations.} An alternative view of this protocol is as solving the following related \emph{amplification} problem:
we are given a signal of strength (rate) $\phi$, and the algorithm's behavior should reflect whether this strength is below or above some threshold.
The detection problem requires us to differentiate thresholds set at $\beta / n$ and $1 / n$, for constant $\beta \ll 1$, but our analysis applies to more general rates.

Above, we have assumed that the leak rate decreases linearly with $n$, to separate from the case where a single instance of $D$ is present.
However, it is also reasonable to consider that the leak rate is \emph{fixed}, say, upper bounded by a constant $\lambda$.
In this case, the analysis works as long as the number of molecules $n$ satisfies $\lambda \ll 1/ n$.

\paragraph{Self-stabilization.}
Our algorithm is self-stabilizing in the sense that if the count of $D$ changes due to some external reason, the output quickly adapts (within logarithmic parallel time). 
This is particularly interesting if the algorithm is used in the context of a ``control module'' for a cell detecting $D$ and the amount of $D$ changes over time.
Note that strawman solutions considered above cannot be ``untriggered'' once $D$ has been detected, and thus cannot adapt to a changing input.

\section{Related Work}
There is much work on attempting to decrease error in the underlying chemical substrate.
A famous example includes kinetic proofreading~\cite{hopfield1974kinetic}.
In the context of DNA strand displacement systems in particular, leak reduction has been a prevailing topic~\cite{thachuk2015leakless}.
Despite the importance of handling leaks, there are few examples of non-trivial \emph{algorithms}, where leaks are handled through computation embedded in chemistry.
One algorithm that appears to be able to handle errors is \emph{approximate majority}~\cite{angluin2008simple}, originally analyzed in a model where a fraction of the nodes are Byzantine, in that they can change their reported state in an adversarial way. 
Potentially due to its robustness properties, the approximate majority algorithm appears to be widely used in biological regulatory networks~\cite{cardelli2014morphisms},
and it was also one of the first chemical reaction network algorithms implemented with strand displacement cascades~\cite{chen2013programmable}.

Our algorithm can be viewed as a timed, self-stabilizing version of rumor spreading. For analysis of simple rumor-spreading, see~\cite{rumourspreading}. Other work include fault-tolerant rumor spreading~\cite{DBLP:journals/dc/DoerrDMM16}, push-pull models~\cite{DBLP:conf/focs/KarpSSV00} and self-stabilizing broadcasting~\cite{DBLP:conf/soda/BoczkowskiKN17}.  A rumor-spreading formulation of the molecule detection problem is also considered in recent work~\cite{DBLP:journals/corr/DudekK17}, which relies on a different source amplification mechanism based on oscillator dynamics. This protocol~\cite{DBLP:journals/corr/DudekK17} is self-stabilizing in a weaker (probabilistic) sense compared to the algorithms from this paper and does not provide leak robustness guarantees.

\section{Preliminaries}
\subsection{Population Protocols with Leaks}

\paragraph{Population Protocols.}
We start from a standard population protocol model, where $n$ molecules (nodes) interact uniformly at random in a series of discrete steps. In our formulation, in each step, a coin is flipped to decide whether the current interaction is a regular reaction or a leak reaction. In the former case, two molecules are picked uniformly at random, and interact according to the rules of the protocol. 
In the latter case, a leak reaction occurs (see below).

A population protocol is given as a set of reactions (transition rules) of the form
$$ A + B \rightarrow C + D,$$
(where some of $A, B, C, D$ might be the same).
We (arbitrarily) match the first reactant ($A$) with the first product ($C$), and the second reactant ($B$) with the second product ($D$), and think of $A$ as changing state to $C$, and $B$ as changing state to $D$.
If the two molecules picked to interact do not have a corresponding interaction rule, then they don't change state and we call this a null interaction.
Population protocols are a special case of the stochastic chemical reaction networks kinetic model (e.g.,~\cite{soloveichik2009robust}(A.4)).

\paragraph{Catalytic and Non-Catalytic Species.}
Given a set of reactions, we define the set of \emph{catalytic} species as the set of states which never change as a consequence of any reaction.
That is, for every reaction, the species is present in the same count both in the input and the output of the reaction. For example, in the reactions
\begin{align*}
A + C \rightarrow B + C \\
A + B \rightarrow A + D
\end{align*}
\noindent we call $C$ \emph{catalytic}. 
Note that $A$ acts as a catalyst in the second reaction, but its count is changed by the first reaction, thus it is not overall catalytic.
All species whose count is modified by some reaction are called non-catalytic.
	Note that it is possible that a species is never created, but disappears as a consequence of an interaction. For example, in the reaction
	$$ L + L \rightarrow A  + B,$$
	
	\noindent $L$ is such as species. We define such species as \emph{non-catalytic}, since their creation is possible by the law of reversibility, and thus they can leak.

\paragraph{An Algorithmic Model of  Leaks.} 
A leak is a reaction of the type
$$ S \rightarrow S'$$
where $S$ and $S'$ are arbitrary non-catalytic species produced by the algorithm.
Note that the input and output species of a leak may be the same (although in that case the reaction is trivial).
In  the following, we make no assumptions on the way in which the input and output of a leak reaction are chosen---we assume that they are chosen adversarially. Instead, we assume an absolute bound on the probability of a leak.

We assume that each reaction is either a \emph{leak reaction} or a \emph{normal reaction}, which follows the algorithm.
We formalize this as follows.
\begin{definition}
	Given an algorithm, defined by a set of reactions, the set of catalysts is the set of species whose count does not change as a consequence of any reaction.
	A \emph{leak} is a spurious reaction, which changes an arbitrary non-catalytic species to an arbitrary non-catalytic species.
	The \emph{leak rate} $\beta / n$ is the probability that any given interaction is a leak reaction.
\end{definition}

\subsection{The Detection Problem}

In the following, we consider the following \emph{detection} task: we are given a distinct species $D$, whose presence or absence must be detected by the algorithm, in the presence of leaks.
 More precisely, if the species $D$ is present, then the algorithm should stabilize to a state in which molecules map to output value ``detect''.
 Otherwise, if $D$ is not present, then the algorithm should stabilize to a state in which molecules map to output value ``non-detect''.
To observe the algorithm's output, we sample a molecule at random, and return its output mapping. (Alternatively, to boost accuracy, we can take a number of  samples, and  return the majority output mapping.)
We require that species $D$ are catalytic.

\section{The Robust-Detect Algorithm}
\label{sec:algo}

\paragraph{Description.}
As given in the problem statement, we assume that there exists a distinguished species $D$, which is to be detected, and which never changes state.
Our algorithm implements a chain of detection species $X_1, \ldots, X_s$, for some parameter $s$, each of which maps to output ``detect'', but with decreasing ``confidence''.
Further, we have a neutral species $N$, which maps to output ``non-detect''.
We assume that the parameter $s = \lceil \log n \rceil$, and that initially all molecules are in state $N$. We specify the transitions below, and provide the intuition behind them.

\vspace{-3mm}
\begin{algorithm}
\label{alg:detect1}
\begin{align*}
	D + X_i  & \rightarrow  D +  X_1, \quad\forall i \in \{  2, \ldots, s  \} \\
	D + N & \rightarrow D + X_1 \\
	X_s + X_s & \rightarrow   N + N \\
	X_s + N & \rightarrow  N + N \\
	X_i + X_j & \rightarrow X_{\min(i, j) + 1} + X_{\min(i, j) + 1}, \quad\forall i, j \in \{  1, 2, \ldots, s - 1  \} \\
	X_i + N & \rightarrow X_{i + 1} + X_{i + 1}, \quad\forall i \in \{ 1, 2, \ldots, s - 1\}\\
\end{align*}
\end{algorithm}
\vspace{-5mm}
The intuition behind the algorithm is as follows.
The ``detecting'' species $X_1, \ldots, X_s$ are arranged in consecutive levels.
Whenever a molecule meets $D$, it moves to the highest ``alert'' level, $X_1$.
Since leaks might produce this species as well, we decay it gracefully across $s$ levels.
After going through these levels, a molecule moves to  neutral state $N$, in case it is not brought back either by meeting $D$, or some molecule at a lower alert level.
For this, whenever two of these species $X_i$ and $X_j$ meet, they both move to level $\min( i, j ) + 1$.
This reaction has the double purpose of both decaying the alert level of the molecule at the lower level, and of bringing back the molecule with the higher alert level.
Further, whenever a molecule at level $X_i$ meets a neutral molecule $N$, it advances its level by $1$.
At the same time, neutral molecules are turned into detector molecules whenever meeting some molecule at an alert level smaller than $s$.

\paragraph{Intuitive Dynamics.} Roughly, the chain of alert levels have the property that, for the first $\sim \log n$ levels, the count roughly \emph{doubles} with level index.
At the same time, past this point, counts exhibit a steep (doubly exponential) drop, so that a small constant fraction of molecules are always neutral.
The presence of $D$ acts like a trigger, which maintains the chain in ``active'' state.
The analysis in the next section makes this intuition precise.  These dynamics are illustrated in Figure~\ref{fig:levels}.

%

\begin{figure}
\centering     
\includegraphics[width=\textwidth]{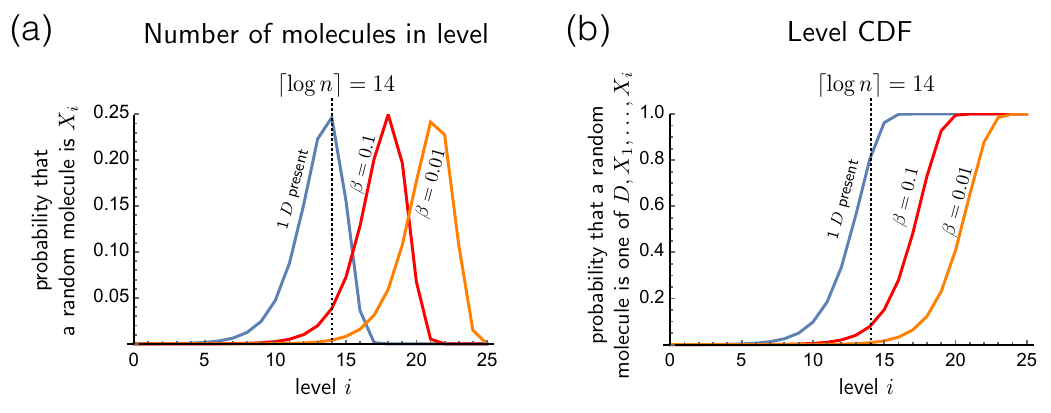}
\caption{Steady state probabilities of the Robust-Detect algorithm for $n = 10^4$ molecules. Three conditions are evaluated: (blue) 1 molecule of $D$ is present and no leak (leak parameter $\beta = 0$); (orange, red) no $D$ is present but with worst-case leak for false-positives (leak reactions $X_i \rxn X_1$ and $N \rxn X_1$) (orange: leak parameter $\beta = 0.01$, red: leak parameter $\beta = 0.1$). 
(a) The probabilities of each level $i$.
(b) The cumulative probabilities of levels $\leq i$, 
capturing the probability that a random molecule is in a ``detect'' state.
Note that it is enough to set the number of levels $s=14 = \lceil \log n \rceil$ to have both false positive and false negative error probabilities small, although for smaller leak rates ($\beta = 0.01$) increasing $s$ beyond $\log n$ can help better distinguish true and false positives.
Numerical probabilities are computed using equations~(\ref{eq:1}) and (\ref{eq:2}).
}
\label{fig:levels}
\end{figure}

\section{Analysis}
\label{sec:analysis}

\paragraph{Overview.} We divide the analysis of the detection algorithm into two parts. First, we derive stationary probabilities of the underlying Markov chain of transitions of particles, by solving recursively the equations following from the underlying dynamics. Later, we derive optimal bounds on the mixing time of this Markov chain---that is we show that probability distribution of states at every time $t \ge c n \log n$ (for some constant $c$) is almost the same as the stationary distribution.

\paragraph{Simplified Algorithm.} For the purpose of analysis, let us consider a following rephrasing of the detection algorithm:
molecule states are $D, X_1, X_2, X_3,\ldots$, and interactions are as follows:

\begin{algorithm}
\label{alg:detect2}
\begin{align*}
D + X_i & \rightarrow D + X_1,\\
X_i + X_j & \rightarrow X_{\min(i,j)+1} + X_{\min(i,j)+1}.\\
\end{align*}
\end{algorithm}
\vspace{-5mm}
This algorithm uses infinite number of states, thus it is useful only for purposes of theoretical analysis. However, it captures the behavior of the original algorithm in the following way:
if in Algorithm~\ref{alg:detect2} all states $X_{s+1},X_{s+2}, \ldots$ are collapsed to $N$, the transitions are equivalent to Algorithm~\ref{alg:detect1}. However, formulation of Algorithm~\ref{alg:detect2} is oblivious to parameter $s$, thus captures simultaneously the dynamics of all possible instances of Algorithm~\ref{alg:detect1}.

\subsection{Stationary Analysis}
Let us consider an initial state when $k \ge 0$ instances of state $D$ are present, with special attention given to $k=0$  and $k=1$. Those molecules do not change their state.

We can imagine tracking a particular molecule through its state transitions,
such that its state can be expressed as a Markov chain.
In the following, we will focus on analyzing the stationary distribution of this Markov chain.

For any $i \in \{1, 2, \ldots, s \}$, we let $p^\star_i$ be the stationary probability that a molecule chosen uniformly at random is in state $X_i$.
We let $p^\star_0 = \frac{k}{n}$ be the (stationary) probability that the molecule is in the state $D$. Let $p^\star_{\le i} = p^\star_0 + \ldots + p^\star_i$ be the probability that a molecule is in any of the states $D, X_1, \ldots, X_i$.

Let us now analyze these stationary probabilities.

\paragraph{Stable State with No Leaks.}
We first analyze the simplified case where no leaks occur.
A molecule $u$ is in one of states $D, X_1, \ldots, X_{i}$ at time $t$, in two cases:
\begin{itemize}
	\item It was in state $D, X_1, \ldots, X_{i}$ at time $t - 1$, and did not get selected for a reaction, which occurs with probability $1 - 2 / n$.
	\item It got selected for a reaction with element $u'$, and either $u$ or $u'$ was in one of states $D, X_1, \ldots, X_{i-1}$.
\end{itemize}

Hence, by stationarity, we get that
\[ p^\star_{\le i} = p^\star_{\le i} \left(1 - \frac{2}{n} \right) + \frac{2}{n} \cdot (1 - (1-p^\star_{\le (i-1)})^2).\]

From this we get that
\[ 1- p^\star_{\le i} = (1-p^\star_{\le (i-1)})^2,\]
which solves to
\begin{align}
p^\star_{\le i} = 1 - \left(1 - \frac{k}{n}\right)^{2^{i}}. \label{eq:1}
\end{align}

This gives us following estimates: if $k\ge 1$, then $p^\star_i \approx 2^{i-1}\frac{k}n$ for $i \le \log (n/k)$. Additionally, for $i = \log(n/k) + 1 + j$, $p^\star_i \approx e^{-2^{j}}$.
Thus, for $i \ge \log(n/k) + \Theta(\log \log n)$ in all practicalities $p^\star_i \approx 0$.

To analyze the probability of detection when $k=1$, we sum probabilities for all $i$ from $0$ to $s = \lceil \log n \rceil$

\[
	\Pr[ \text{detect} ] = \sum_{i = 0}^{s} p^\star_i = p^\star_{\le s} \ge  1 - \left(1 - \frac1n\right)^n  \ge 1 - \frac1e.
\]

\medskip
\paragraph{Probability of False Positives with Leaks.} A useful side effect of the previous analysis is that we also get probability bounds for detection in the case where $D$ is \emph{not} present, i.e. false positives.
We model this case as follows. Assume that there exists an upper bound $\lambda$ on the probability that a certain reaction is a leak.
Examining the structure of the algorithm, we note that the worst-case adversarial application of leaks would be if this probability is entirely concentrated into leaks which produce species $X_1$.

To preserve molecular count, we assume the following simplified leak model, which is equivalent to the general one, but easier to deal with in the confines of our algorithm.

Each reaction is a \emph{leak} with probability $\lambda = \beta / n$, where $\beta \ll 1$ is a small constant.
If a reaction is a leak, it selects a molecule at random, and transforms it into an arbitrary state. In this case, we will assume adversarially that all leaked molecules are transformed into state $X_1$.
Notice that the assumption that $\beta \ll 1$ is required to separate this setting from the case where $D$ is present in the system, where the probability of producing state $X_1$ is $2 / n$.

We continue with calculations of $p^\star_0,p^\star_1,\ldots$ for the above formulation.
Note that the recurrence relation for $p^\star_{\le i}, i \geq 1$ is changed as follows:
\begin{itemize}
\item If at that round there was no leak, the transition probabilities are as previously. This happens with probability $1 - \frac{\beta}{n}$.
\item If there was a leak, then the molecule either is selected as a leaked molecule (this happens with probability $\frac{1}{n} \cdot \frac{\beta}{n}$) or it was not selected as a leaked molecule, but it was already in the proper state (probability $\frac{\beta}{n} \cdot \frac{n-1}{n} p^\star_{\le i}$).
\end{itemize}
The recursive formulation gives
\[
	p^\star_{\le i} = \left(p^\star_{\le i} \left(1 - \frac2n\right) + \frac2n \left(1 - (1-p^\star_{\le(i-1)})^2\right) \right)\left(1-\frac\beta{n}\right) + \left(\frac1n + \frac{n-1}{n}p^\star_{\le i} \right)\frac{\beta}{n}.
\]
Which is equivalent to
\[1 - p^\star_{\le i} = \frac{\left(1-\frac\beta{n}\right)}{\left(1-\frac\beta{2n}\right)}(1-p^\star_{\le(i-1)})^2\]
leading to (using estimate $(1-\frac\beta{n})/(1-\frac\beta{2n}) \approx (1-\frac\beta{2n})$)
\begin{align}
p^\star_{\le i} \approx 1 - \left(1 - \frac{\beta}{2n}\right)^{1+2+\ldots+2^{i-1}} = 1 - \left(1 - \frac{\beta}{2n}\right)^{2^{i}-1} .   \label{eq:2}
\end{align}

This gives us following estimates: $p^\star_i \approx 2^{i-2}\frac{\beta}n$ for $i \le \log (2n/\beta)$. Additionally, for $i = \log(2n/\beta) + 1 + j$, $p^\star_i \approx e^{-2^{j}}$.
Thus, for $i \ge \log(2n/\beta) + \Theta(\log \log n)$ in all practicalities $p^\star_i \approx 0$.

This immediately implies that
\[
	\Pr[\text{detect}] = \sum_{i=0}^{s} p^\star_{i}  = p^\star_{\le s} \le 1 - \left(1 - \frac{\beta}{2n}\right)^{2n} = 1 - \frac{1}{e^{\beta}} \approx \beta,
\]
which means that the probability that a randomly chosen molecule is in detect state when chosen uniformly at random is at most $\beta$.

\paragraph{Probability of False Negatives with Leaks.}
Under the same leak model, it is easy to notice that the ``best'' adversarial strategy for our algorithm in case $D$ is present is to concentrate all leaks to create the neutral species $N$ (or $X_{\infty}$ in case of Algorithm~\ref{alg:detect2}).
It is easy to see that this just decreases the total probability of detect states by the leak probability $\lambda = \beta / n$. More formally, we compute once again stationary probabilities. The recurrent relation is

\[ p^\star_{\le i} = \left(p^\star_{\le i} \left(1 - \frac{2}{n} \right) + \frac{2}{n} \cdot (1 - (1-p^\star_{\le (i-1)})^2)\right)\left(1 - \frac\beta{n}\right) + \frac{\beta}{n} \cdot \frac{n-1}{n} p^\star_{\le i}.\]
Using estimate $(1-\frac\beta{n})/(1-\frac\beta{2n}) \approx (1-\frac\beta{2n})$ we reach
\[p^\star_{\le i} = \left(1 - \frac{\beta}{2n}\right) (1 - (1-p^\star_{\le (i-1)})^2).\]
Thus we have for the first $\log (n/k)$ levels the dampening factor of $(1 - \beta/(2n))$ per level (compared to the leakless case). It can be easily shown by induction that
\[ \left( 1 - \frac{\beta}{2n} \right)^i \left(1 - \frac{k}{n}\right)^{2^i} \le p^\star_{\le i} \le \left(1 - \frac{k}{n}\right)^{2^i}. \]

The estimates for $p^\star_{i}$ follow from the leakless case, after taking into the account the composed dampening factor:
\[
	\Pr[\text{detect}] = \sum_{i=0}^{s} p^\star_{i}  = p^\star_{\le s} \ge \left(1 - \frac{1}{e}\right)\cdot\left(1-\frac{\beta}{2n}\right)^{\log n}  =  1 - \frac1e - \bigo\left(\frac{ \log n}{ n}\beta\right).
\]

Finally, we summarize the results in this section as follows:
\begin{theorem}
	Assuming leak rate $\beta / n$ for $\beta \ll 1$, Robust-Detect guarantees the following.
	\begin{itemize}
		\item The probability of a false positive is at most $\beta$.
		\item The probability of a false negative is at most $1/e + \bigo(\beta \cdot (\log n)/n) $.
	\end{itemize}
\end{theorem}

Notice that these probabilities can be boosted by standard sampling techniques.

\subsection{Convergence Analysis}
We now proceed with an analysis of the convergence speed of the previously described protocols. To avoid separate analysis for each of the aforementioned cases (no leaks, false positives, false negatives) and to be independent from all possible initializations of the algorithm, we first start with showing that, under no leaks and with no $D$ present, all states $X_1, \dots, X_c$ are quickly killed.

In this section, it is more convenient to use $t$ to refer to the total number of interactions, rather than parallel time. 
To convert to parallel time, one needs to divide by $n$, the number of molecules.

\begin{lemma}
\label{lem:cleaning}
Assume arbitrary (adversarial) initial state in $t=0$ and evolution with no leaks ($\beta = 0$) and no $D$ is present. For any $c(n) \ge 1$, there is $t = \bigo(n \cdot (c(n) + \log n))$ such that with probability $1-1/n^{\Theta(1)}$ (with high probability) there is no molecule in any of the states $X_1, X_2, \ldots, X_{c(n)}$ after $t$ interactions.
\end{lemma}
\begin{proof}
We assign a potential to each molecule, based on the state it is currently in:
$\Phi(X_i) = 3^{-i}$. We also define a global potential $\Phi_t$ as sum of all molecular potentials after $t$ interactions.
Observe, that when two molecules interact, following rule
$X_i + X_j \rightarrow X_{\min(i, j) + 1} + X_{\min(i, j) + 1}$, then there is:
\[\Phi( X_{\min(i, j) + 1}) + \Phi( X_{\min(i,j) + 1} ) \le 2/3 \cdot \left(\Phi(X_i) + \Phi(X_j)\right),\]
which can be interpreted that each interacting molecule loses at least $1/3$ of its potential.
Since each molecule participates in an interaction with probability $\frac{2}{n}$ in each round, the following bound holds:
\[ \mathbf{E}[\Phi_{t}-\Phi_{t+1} | \Phi_t] \ge \cdot \sum_{v} \Pr(v \text{ interacts in round } t)\cdot \frac{1}{3}\Phi_t(v) = \frac{2}{3n} \Phi_t,\]
\[\mathbf{E}[\Phi_{t+1} | \Phi_t] \le \left(1 - \frac{2}{3n}\right) \Phi_t.\]
Substituting $\Phi_0 \le n$ and fixing $t \ge \frac{3}{2}n \ln (n \cdot 3^{c(n)} \cdot n^{\Theta(1)}) = \bigo(n (c(n) + \log n + \Theta(\log n)))$ we have
\[\mathbf{E}[\Phi_t]\le  \left(1 - \frac{2}{3n}\right)^t \cdot n \le e^{-\ln (n \cdot 3^{c(n)} \cdot n^{\Theta(1)})} \cdot n = 3^{-c(n)} \cdot \frac{1}{n^{\Theta(1)}}.\]
By Markov's inequality, this means that there is no molecule in any of the states $X_1, X_2, \ldots, X_c$ with probability at least $1 - n^{-\Theta(1)}$, that is with high probability.
\qed
\end{proof}

We mention one additional useful property of Algorithm~\ref{alg:detect2}, that its actions on population are decomposable with respect to levels. That is, define $\textsc{level}_t(u) = i$ if molecule $u$ at time $t$ is in state $X_i$, and $\textsc{level}_t(u) = 0$ if it is in state $D$.
\begin{observation}
\label{obs:coupling}
Let $\{u_1,u_2,\ldots,u_n\},\{v_1,v_2,\ldots,v_n\},\{w_1,w_2,\ldots,w_n\}$ be 3 disjoint populations each on $n$ molecules, following evolution defined by Algorithm~\ref{alg:detect2}. Moreover, let their evolutions be coupled: at each time $t$, in each population the corresponding molecules interact (i.e., the interaction is $u_i + u_j$, $v_i + v_j$, $w_i + w_j$ in the three populations for some $i, j$).

If $\forall_i \textsc{level}_0(u_i) = \min(\textsc{level}_0(v_i),\textsc{level}_0(w_i)),$
then at any time $t>0:$
$\forall_i \textsc{level}_t(u_i) = \min(\textsc{level}_t(v_i),\textsc{level}_t(w_i)).$
\end{observation}
This observation can be naturally generalized to more than 3 populations.
As shown below, the observation implies that to analyze detection under noisy start, we can decouple starting noise from detected particle and analyze evolution under those two separately.
Denote by $p_i(t)$ and $p_{\le i}(t)$ the probability for a randomly picked molecule after $t$ interactions to be in the state $X_i$ or $D,X_1,\ldots,X_i$ respectively.

\begin{theorem}
Fix arbitrary leak model (i.e. no leaks, false-positives, false-negatives)
and arbitrary concentration of $D$. For any $c \ge 1$, and $t = \Omega(n \cdot (c + \log n))$, there is
\[ \left| p^\star_{\le c} - p_{\le c}(t) \right| \le 1/n^{\Theta(1)}, \]
where $p^\star$ is the stationary probability distribution of the identical process.
\end{theorem}
\begin{proof}
First, for simplicity we collapse all states $X_{c+1}, X_{c+2}, \ldots$ into $N$, since it has no effect on $p_{\le c}$ distributions.
Consider a population of size $n$, under no leaks, no $D$, evolution. By Lemma~\ref{lem:cleaning}, in $\tau = \bigo(n \cdot (c + \log n))$ steps it reaches all-$N$ state, regardless of initial configuration, with high probability.
Thus evolution of any population $\{u_i\}$, under no leaks, with $D$ present, is a coupling (as in Observation~\ref{obs:coupling}) of following evolutions:
\begin{itemize}
\item initial configuration of population $\{u_i\}$, with each $D$ replaced with $N$;
\item for every timestep $t_i$ such that $D$ interacted with $X_i$ or $N$ creating $X_1$, we couple a population with corresponding molecule set to $X_1$ and every other molecule set to $N$, shifted in time so its evolution starts at time $t_i$.
\end{itemize}
Observe, that evolution of population of both types will reach all-$N$ state in $\tau$ steps, with high probability. Thus, conditioned on this high probability, the configuration at any $t \ge \tau$ is the result of coupling of all-$N$ (result of evolution of first type population) with possibly several configurations of the second type, where at each timestep $t' \in [t-\tau,t]$ such population was created independently with some probability only depending on $n$ and $k$. 
However, the coupling we just described is invariant from the choice of $t$, as long as $t \ge \tau$. Thus, for any $t_1,t_2 \ge \tau$, there is $\left| p_{\le c}(t_1) - p_{\le c}(t_2) \right| \le 1/n^{\Theta(1)}$. Since $p^\star_{\le c} = \lim_{t \to \infty} \frac{1}{t}\sum_{i=1}^t p_{\le c}(i)$, the claimed bound follows.

To take into account errors, we say that whenever there is a leak changing state of molecule $v$ to some state $S$ at time $t$, we change state of $v$ at that time in all existing populations to $N$, and create new population where $v$ has state $S$, and all other molecules are in $N$ state. The same reasoning as in the error-less case follows, since switching molecules to $N$ state it only speeds up convergence of populations to all-$N$ state and since  populations created due to leaks are created at each step with the same probability depending only on $n$ and error model.
\qed
\end{proof}

\section{Simulation Results}
\label{sec:simulation}

We simulated the Robust-Detect algorithm (Algorithm~\ref{alg:detect1}) using a modified version of the CRNSimulatorSSA Mathematica package~\cite{crnsimulator}.
Figure~\ref{fig:sims} shows the shape of typical trajectories when there is one molecule in state $D$ ($k = 1$), compared with no molecules in state $D$ ($k=0$) but with the worst-case leak for false-positives.
Note that $D$ is quickly detected if present, 
and if absent the system exhibits random perturbations that are quickly extinguished and are clearly distinguishable from the true positive case.

\begin{figure}[t]
\begin{center}
\includegraphics[width=\textwidth]{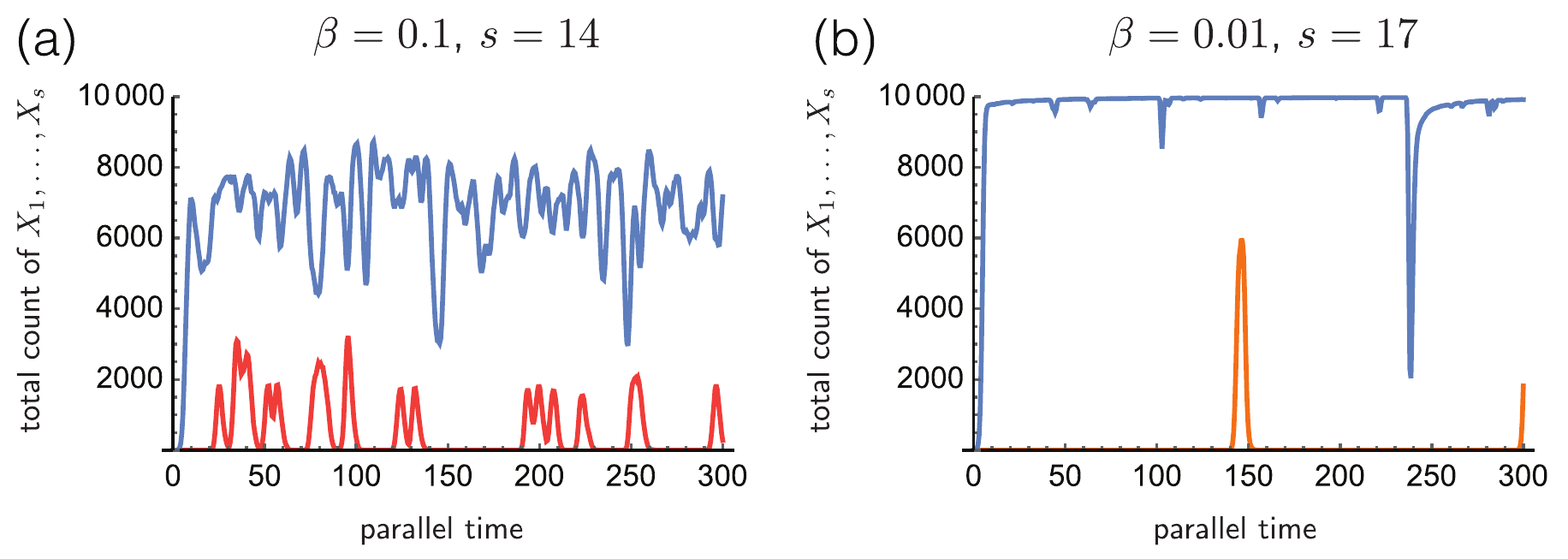}
\caption{
Typical time-evolution of the Robust-Detect algorithm (Algorithm~\ref{alg:detect1}).
Three colors correspond to the three conditions described in Fig.~\ref{fig:levels}:
(blue) 1 molecule of $D$ is present and no leak (leak parameter $\beta = 0$); (orange, red) no $D$ is present but with worst-case false-positive leak $X_i \rxn X_1$ and $N \rxn X_1$ (orange: leak parameter $\beta = 0.01$, red: leak parameter $\beta = 0.1$).
All $X_i$ states map to output value ``detect'', and thus we plot the sum of all their counts.
(a) $s = 14$ layers, $\beta = 0.1$ (red).
(b) $s = 17$ layers, $\beta = 0.01$ (orange). 
See Fig.~\ref{fig:levels} for the corresponding steady state probabilities.
Note that with smaller leak ($\beta = 0.01$), it is possible to better distinguish true positives and false positives by increasing the number of layers (from $14$ to $17$). 
In all cases there are $n = 10^4$ molecules, and the initial configuration is all molecules in neutral state $N$.
Parallel time (number of interactions divided by $n$) corresponds to the natural model of time where each molecule can interact with a constant number of other molecules per unit time.}
\label{fig:sims}
\end{center}
\end{figure}


\section{Conclusions}
We have considered the problem of modeling and withstanding leaks in chemical reaction networks, expressed as population protocols.
We have presented an arguably simple algorithm which is probabilistically correct under assumptions on the leak rate, and converges quickly to the correct answer.

Beyond the specific example of robust detection, we hope that our results motivate more systematic modeling of leaks, and further work on algorithmic techniques to withstand them. As such errors appear to spring from the basic laws of chemistry, their explicit treatment appears to be necessary.
The authors found it surprising that many of the algorithmic techniques developed in the context of deterministically correct population protocols might not carry over to implementations, due to their inherent non-robustness to leaks.

In future work, we plan to perform an exhaustive examination of which of the current algorithmic techniques could be rendered leak-robust, and whether known algorithms can be modified to withstand leaks via new techniques. Another interesting avenue for future work is \emph{lower bounds} on the set of computability or complexity of fundamental predicates in the leak model. Finally, we would like to examine whether our robust detection algorithm can be implemented in strand displacement systems.

\vspace{5pt}
\paragraph{Acknowledgments.}
We thank Lucas Boczkowski and Luca Cardelli for helpful comments on the manuscript.

\bibliography{MPP,refs}

\begin{thebibliography}{10}
\providecommand{\url}[1]{#1}
\csname url@samestyle\endcsname
\providecommand{\newblock}{\relax}
\providecommand{\bibinfo}[2]{#2}
\providecommand{\BIBentrySTDinterwordspacing}{\spaceskip=0pt\relax}
\providecommand{\BIBentryALTinterwordstretchfactor}{4}
\providecommand{\BIBentryALTinterwordspacing}{\spaceskip=\fontdimen2\font plus
\BIBentryALTinterwordstretchfactor\fontdimen3\font minus
  \fontdimen4\font\relax}
\providecommand{\BIBforeignlanguage}[2]{{%
\expandafter\ifx\csname l@#1\endcsname\relax
\typeout{** WARNING: IEEEtran.bst: No hyphenation pattern has been}%
\typeout{** loaded for the language `#1'. Using the pattern for}%
\typeout{** the default language instead.}%
\else
\language=\csname l@#1\endcsname
\fi
#2}}
\providecommand{\BIBdecl}{\relax}
\BIBdecl

\bibitem{gopalkrishnan2011catalysis}
M.~Gopalkrishnan, ``Catalysis in reaction networks,'' \emph{Bulletin of
  mathematical biology}, vol.~73, no.~12, pp. 2962--2982, 2011.

\bibitem{soloveichik2010dna}
D.~Soloveichik, G.~Seelig, and E.~Winfree, ``{DNA} as a universal substrate for
  chemical kinetics,'' \emph{Proceedings of the National Academy of Sciences},
  vol. 107, no.~12, pp. 5393--5398, 2010.

\bibitem{cardelli2013two}
L.~Cardelli, ``Two-domain {DNA} strand displacement,'' \emph{Mathematical
  Structures in Computer Science}, vol.~23, no.~02, pp. 247--271, 2013.

\bibitem{chen2013programmable}
Y.-J. Chen, N.~Dalchau, N.~Srinivas, A.~Phillips, L.~Cardelli, D.~Soloveichik,
  and G.~Seelig, ``Programmable chemical controllers made from {DNA},''
  \emph{Nature Nanotechnology}, vol.~8, no.~10, pp. 755--762, 2013.

\bibitem{zhang2007engineering}
D.~Y. Zhang, A.~J. Turberfield, B.~Yurke, and E.~Winfree, ``Engineering
  entropy-driven reactions and networks catalyzed by {DNA},'' \emph{Science},
  vol. 318, no. 5853, pp. 1121--1125, 2007.

\bibitem{angluin2006passivelymobile}
D.~Angluin, J.~Aspnes, Z.~Diamadi, M.~Fischer, and R.~Peralta, ``Computation in
  networks of passively mobile finite-state sensors,'' \emph{Distributed
  Computing}, vol.~18, pp. 235--253, 2006.

\bibitem{DBLP:conf/focs/KarpSSV00}
R.~M. Karp, C.~Schindelhauer, S.~Shenker, and B.~V{\"{o}}cking, ``Randomized
  rumor spreading,'' in \emph{{FOCS}}, 2000, pp. 565--574.

\bibitem{thachuk2015leakless}
C.~Thachuk, E.~Winfree, and D.~Soloveichik, ``Leakless {DNA} strand
  displacement systems,'' in \emph{{DNA} Computing and Molecular
  Programming}.\hskip 1em plus 0.5em minus 0.4em\relax Springer, 2015, pp.
  133--153.

\bibitem{angluin2008simple}
D.~Angluin, J.~Aspnes, and D.~Eisenstat, ``A simple population protocol for
  fast robust approximate majority,'' \emph{Distributed Computing}, vol.~21,
  no.~2, pp. 87--102, 2008.

\bibitem{cardelli2014morphisms}
L.~Cardelli, ``Morphisms of reaction networks that couple structure to
  function,'' \emph{BMC Systems Biology}, vol.~8, no.~1, p.~84, 2014.

\bibitem{rumourspreading}
B.~Pittel, ``On spreading a rumor,'' \emph{SIAM Journal on Applied
  Mathematics}, vol.~47, no.~1, pp. 213--223, 1987.

\bibitem{DBLP:journals/dc/DoerrDMM16}
B.~Doerr, C.~Doerr, S.~Moran, and S.~Moran, ``Simple and optimal randomized
  fault-tolerant rumor spreading,'' \emph{Distributed Computing}, vol.~29,
  no.~2, pp. 89--104, 2016.

\bibitem{DBLP:conf/soda/BoczkowskiKN17}
L.~Boczkowski, A.~Korman, and E.~Natale, ``Minimizing message size in
  stochastic communication patterns: Fast self-stabilizing protocols with 3
  bits,'' in \emph{{SODA}}, 2017, pp. 2540--2559.

\bibitem{DBLP:conf/stoc/DudekK18}
B.~Dudek and A.~Kosowski, ``Universal protocols for information dissemination
  using emergent signals,'' in \emph{{STOC}}, 2018, pp. 87--99.

\bibitem{soloveichik2009robust}
D.~Soloveichik, ``{Robust stochastic chemical reaction networks and bounded
  tau-leaping},'' \emph{Journal of Computational Biology}, vol.~16, no.~3, pp.
  501--522, 2009.

\bibitem{crnsimulator}
\url{http://users.ece.utexas.edu/~soloveichik/crnsimulator.html}.

\end{thebibliography}
\bibliographystyle{IEEEtran}

\end{document}